\documentclass[conference,letterpaper]{IEEEtran}

\usepackage[utf8]{inputenc} 
\usepackage[T1]{fontenc}
\usepackage{url}
\usepackage{ifthen}
\usepackage{cite}
\usepackage[hidelinks,breaklinks=true]{hyperref}
\usepackage[cmex10]{amsmath} % Use the [cmex10] option to ensure complicance
\usepackage{url}

\usepackage[utf8]{inputenc}
\usepackage[T1]{fontenc}
\usepackage{url}
\usepackage{ifthen}
\usepackage{cite}
\usepackage[cmex10]{amsmath}
\usepackage{amssymb,amsfonts}
\usepackage{graphicx}
\usepackage{textcomp}
\usepackage{xcolor}
\usepackage{makecell}
\usepackage{algorithmic}
\usepackage{comment}
\usepackage{tikz}
\newenvironment{remark}
{\par\noindent\textit{Remark.}\ }
{\par}
\usetikzlibrary{arrows.meta,shapes.geometric,arrows,positioning,calc}

\newtheorem{theorem}{Theorem}

\begin{document}
\title{Real-Time Reconstruction and Actuation Error Analysis for Markov Sources over MPR Channels}
% %%% Single author, or several authors with same affiliation:
% \author{%
%  \IEEEauthorblockN{Author 1 and Author 2}
% \IEEEauthorblockA{Department of Statistics and Data Science\\
%                    University 1\\
 %                   City 1\\
  %                  Email: author1@university1.edu}% }

%%% Several authors with up to three affiliations:
\author{%
\IEEEauthorblockN{Pansee S. Elessawy and Nikolaos Pappas}
\IEEEauthorblockA{Department of Computer and Information Science, 
Linköping University, Linköping, Sweden\\
Email: \{pansee.elessawy, nikolaos.pappas\}@liu.se}
}

\maketitle

\begin{abstract}

We study real-time reconstruction and actuation for two binary Markov sources that share a wireless multi-packet reception (MPR) channel. Each sensor follows a stationary randomized sampling policy, and the receiver maintains source estimates using the most recently decoded updates. We derive closed-form expressions for the steady-state real-time reconstruction error (RTE) and the cost of actuation error (CAE) as functions of the source transition probabilities and the effective update probabilities. We then characterize these update probabilities under randomized sampling, linking the physical-layer MPR model to task-oriented reconstruction and actuation metrics. Using these expressions, we formulate a sampling-constrained optimization problem with a weighted-error objective. The resulting analysis reveals how source dynamics, semantic weights, and 
MPR coupling affect the allocation of sampling resources.
Numerical results show that optimized randomized sampling outperforms random, greedy, and time-sharing baselines.
\end{abstract}

\section{Introduction}
\label{sec:introduction}

Networked autonomous systems, such as industrial automation, connected robotics, and cyber-physical control platforms, rely on timely and accurate
information exchange for monitoring and actuation
\cite{vitturi2013industrial,gielis2022critical}. In such systems, communication should be assessed not only by packet reliability or throughput, but also by its impact on the underlying task \cite{PappasKountouris2021ICAS, LuoDelfaniSalimnejadPappas2025,popovski2020semantic,CALVANESESTRINATI2021107930,9919752}.

The Age of Information (AoI) has been widely used to quantify the freshness of status updates at a remote destination \cite{KaulYatesGruteser2012,KostaPappasAngelakis2017,
YatesSunBrownKaulModianoUlukus2021,8000687}. However, freshness alone does not capture whether the receiver's estimate is correct, nor the cost incurred when actuation is based on an incorrect estimate. This has motivated beyond-AoI and goal-oriented
metrics, such as the Age of Incorrect Information (AoII) \cite{MaatoukKriouileAssaadEphremides2020,11071330,9921185,9162726,10620879}, the real-time reconstruction error (RTE), and the cost of actuation error (CAE) \cite{10409276,LuoDelfaniSalimnejadPappas2025}. These metrics evaluate information according to its usefulness for real-time reconstruction and actuation.

Goal-oriented sampling for Markov sources has been studied in several settings. Semantics-aware sampling for real-time tracking over an unreliable point-to-point channel was introduced in  \cite{PappasKountouris2021ICAS}, while CAE minimization under sampling constraints was considered in \cite{FountoulakisPappasKountouris2023}. Fundamental limits of remote estimation under communication constraints were studied in
\cite{ChakravortyMahajan2017,8812616}. In addition, goal-oriented estimation of multiple Markov sources in constrained networks was studied in \cite{LuoPappas2025,11450399,10547339} and pull-based remote tracking with correlated observations or partial observations in
~\cite{ZakeriMoltafetCodreanu2025,11195539,9676636}. 
AoI has been analyzed in multiple-access systems\cite{
8006544} with multi-packet reception (MPR) \cite{9365698}, and unreliable wireless networks with scheduling for weighted AoI minimization \cite{KadotaSinhaUysalSinghModiano2018,8734015}. Recent works have studied status updating over ALOHA-based random access channels \cite{9358219,10670063}. The works ~\cite{10949089,dlr224073} studied two-state Markov source monitoring including joint model estimation with unknown source statistics and analytical renewal--reward characterizations under random and reactive access policies.

In this paper, we study two independent binary Markov sources that share a wireless multiple-access channel with MPR. Unlike collision-based access, MPR allows the receiver to decode more than one simultaneous transmission with nonzero probability \cite{GhezVerduSchwartz1988}. The update probability of each source depends not only on its own sampling decisions, but also on the sampling decisions of the other sensor. 
\textit{While prior work has mainly studied goal-oriented sampling over point-to-point links or AoI in random access, the role of MPR in real-time reconstruction and actuation has not been thoroughly investigated}.

The main contributions of this paper are as follows. We derive closed-form expressions for the steady-state RTE of binary Markov sources under synchronize-or-hold estimation. We characterize the effective update probabilities induced by stationary randomized sampling over an MPR channel, thereby linking the MPR reception model to the reconstruction and actuation metrics. We obtain the CAE and show that, for binary sources, CAE minimization is equivalent to weighted RTE minimization with modified source weights. Finally, we formulate a sampling-constrained weighted-error minimization problem and study how source dynamics, semantic weights, and MPR coupling affect the allocation of sampling resources.
\section{System Model}
\label{sec:system_model}

We consider a time-slotted remote monitoring system with two sensors, two independent binary Markov sources, and a common receiver communicating over a shared wireless channel with multi-packet reception (MPR), as shown in Fig.~\ref{fig:one}. Time is indexed by $t\in\{0,1,2,\ldots\}$. The state of
source $i\in\{1,2\}$ at slot $t$ is denoted by $X_i(t)\in\{0,1\}$ and evolves according to the transition matrix
\begin{equation}
\mathbf{P}_i =
\begin{bmatrix}
1-\alpha_i & \alpha_i \\
\beta_i & 1-\beta_i
\end{bmatrix},
\qquad \alpha_i,\beta_i\in(0,1),
\label{eq:sys_Pi}
\end{equation}

where $\alpha_i$ and $\beta_i$ are the transition probabilities from $0$ to $1$ and from $1$ to $0$, respectively.

At each slot, sensor $k\in\{1,2\}$ selects an action
\begin{equation}
a_k(t)\in\{0,1,2\},
\label{eq:sys_action_space}
\end{equation}
where $a_k(t)=0$ denotes silence, while $a_k(t)=i$ means that sensor $k$ samples and transmits the current state of source $i$. We consider the stationary randomized sampling policy
\begin{equation}
\Pr\{a_k(t)=j\}=a_{k,j},\qquad j\in\{0,1,2\},
\label{eq:sys_random_policy_prob}
\end{equation}
where
\begin{equation}
a_{k,0}+a_{k,1}+a_{k,2}=1,\qquad a_{k,j}\geq 0.
\label{eq:sys_random_policy}
\end{equation}
The policy of sensor $k$ is denoted by $\mathbf{a}_k\triangleq(a_{k,0},a_{k,1},a_{k,2})$

and the joint action by $\mathbf{a}(t)\triangleq(a_1(t),a_2(t))$.

Each transmission occupies one slot. Let $Z_k(t)\in\{0,1\}$ denote the decoding outcome of sensor $k$, where $Z_k(t)=1$ if the packet from sensor $k$ is successfully decoded, otherwise the packet is discarded. 
The MPR channel is characterized by the success probabilities
$p_{1/1}$ and $p_{2/2}$ when sensors $1$ and $2$ transmit alone, respectively, and by $p_{1/1,2}$ and $p_{2/2,1}$ when both sensors transmit simultaneously. We have $p_{1/1}\geq p_{1/1,2}, p_{2/2}\geq p_{2/2,1}$.\footnote{The MPR success probabilities are used as a physical-layer abstraction of the wireless transmission process. In a specific wireless model, these probabilities can be related to channel parameters such as fading statistics, path loss, interference, and decoding threshold~\cite{NawareMergenTong2005}.} 

The receiver maintains an estimate $\hat X_i(t)\in\{0,1\}$ of each source. A successful update for source $i$ occurs if at least one successfully decoded packet carries $X_i(t)$. Thus,
\begin{equation}
\begin{aligned}
U_i(t)\triangleq
\mathbf{1}\Big\{&
(a_1(t)=i, Z_1(t)=1)
\ \text{or}\
(a_2(t)=i, Z_2(t)=1)
\Big\},\\
& i\in\{1,2\}.
\end{aligned}
\label{eq:sys_Ui}
\end{equation}
The receiver follows a synchronize-or-hold rule
\begin{equation}
\hat X_i(t)=
\begin{cases}
X_i(t), & U_i(t)=1,\\
\hat X_i(t-1), & U_i(t)=0,
\end{cases}
\qquad i\in\{1,2\}.
\label{eq:sys_sync_or_hold}
\end{equation}
The effective update probability of source $i$ is defined as
\begin{equation}
q_i\triangleq \Pr\{U_i(t)=1\},\qquad i\in\{1,2\}.
\label{eq:sys_qi_def}
\end{equation}

\begin{figure}[t]
\centering
\resizebox{0.88\linewidth}{!}{%
\begin{tikzpicture}[
    >=Stealth,
    font=\footnotesize,
    mc/.style={circle, draw, minimum size=0.42cm, inner sep=0pt, font=\scriptsize},
    sensor/.style={
        circle, draw, thick, fill=green!8,
        minimum size=1.15cm, align=center,
        inner sep=0pt, font=\scriptsize
    },
    mpr/.style={
        rectangle, draw, thick, fill=blue!5,
        rounded corners=2pt,
        minimum width=2.6cm, minimum height=0.95cm,
        align=center, font=\scriptsize
    },
    receiver/.style={
        circle, draw, thick,
        minimum size=1.15cm, align=center,
        inner sep=0pt, font=\scriptsize
    },
    estbox/.style={
        rectangle, draw, thick, fill=cyan!5,
        minimum width=1.7cm, minimum height=0.85cm,
        align=center, inner sep=1.5pt, font=\scriptsize
    },
    elab/.style={font=\scriptsize, inner sep=1pt}
]

% ============== GRAY SOURCE BOXES (drawn first, behind everything) ==============
\draw[thick, fill=gray!15, rounded corners=1.5pt]
    (-3.90, 4.35) rectangle (-0.25, 6.40);
\draw[thick, fill=gray!15, rounded corners=1.5pt]
    ( 0.25, 4.35) rectangle ( 3.90, 6.40);

% Anchor coordinates on box bottoms
\coordinate (box1S) at (-2.05, 4.35);
\coordinate (box2S) at ( 2.05, 4.35);

% ============== SOURCE 1 ==============
\node[elab] at (-2.05, 6.05) {\textbf{Source $X_1(t)$}};
\node[mc] (a0) at (-2.55, 5.20) {0};
\node[mc] (a1) at (-1.55, 5.20) {1};
\path[-{Stealth[scale=0.7]}]
    (a0) edge[loop left, looseness=6]
        node[elab, left] {$1{-}\alpha_1$} (a0)
    (a0) edge[bend left=22]
        node[elab, above, inner sep=1pt] {$\alpha_1$} (a1)
    (a1) edge[bend left=22]
        node[elab, below, inner sep=1pt] {$\beta_1$} (a0)
    (a1) edge[loop right, looseness=6]
        node[elab, right] {$1{-}\beta_1$} (a1);

% ============== SOURCE 2 ==============
\node[elab] at (2.05, 6.05) {\textbf{Source $X_2(t)$}};
\node[mc] (b0) at (1.55, 5.20) {0};
\node[mc] (b1) at (2.55, 5.20) {1};
\path[-{Stealth[scale=0.7]}]
    (b0) edge[loop left, looseness=6]
        node[elab, left] {$1{-}\alpha_2$} (b0)
    (b0) edge[bend left=22]
        node[elab, above, inner sep=1pt] {$\alpha_2$} (b1)
    (b1) edge[bend left=22]
        node[elab, below, inner sep=1pt] {$\beta_2$} (b0)
    (b1) edge[loop right, looseness=6]
        node[elab, right] {$1{-}\beta_2$} (b1);

% ============== SENSORS / MPR / RECEIVER / ESTIMATES ==============
\node[sensor] (sen1) at (-2.05, 2.85) {\textbf{Sensor 1}\\policy $\mathbf{a}_1$};
\node[sensor] (sen2) at ( 2.05, 2.85) {\textbf{Sensor 2}\\policy $\mathbf{a}_2$};

\node[mpr] (mpr) at (0, 1.20) {\textbf{MPR Channel}\\
    \scriptsize $p_{1/1},\, p_{2/2},\, p_{1/1,2},\, p_{2/2,1}$};

\node[receiver] (rec) at (0, -0.40) {\textbf{Receiver}\\$R$};

\node[estbox] (hat1) at (-1.5, -2.05) {$\hat{X}_1(t)$\\\scriptsize Estimate};
\node[estbox] (hat2) at ( 1.5, -2.05) {$\hat{X}_2(t)$\\\scriptsize Estimate};

% ============== ARROWS ==============
% Direct (own-source) sampling
\draw[->, thick] (box1S) -- (sen1.north)
    node[midway, left=1pt, font=\scriptsize] {$a_{1,1}$};
\draw[->, thick] (box2S) -- (sen2.north)
    node[midway, right=1pt, font=\scriptsize] {$a_{2,2}$};

% Cross sampling
\draw[->, thick]
    (-1.20, 4.35) .. controls +(0.6,-0.7) and +(-0.5,0.7) .. (sen2.north west)
    node[pos=0.1, below, sloped, font=\scriptsize] {$a_{1,2}$};
\draw[->, thick]
    ( 1.20, 4.35) .. controls +(-0.6,-0.7) and +(0.5,0.7) .. (sen1.north east)
    node[pos=0.1, below, sloped, font=\scriptsize] {$a_{2,1}$};

% Sensors -> MPR
\draw[->, thick] (sen1.south) -- ($(mpr.north)+(-0.05,0.01)$);
\draw[->, thick] (sen2.south) -- ($(mpr.north)+( 0.05,0.01)$);

% MPR -> Receiver
\draw[->, thick, dashed] (mpr.south) -- (rec.north);

% Receiver -> Estimates
%\draw[->, thick] (rec.south west) -- (hat1.north east)
%    node[midway, left=1pt, font=\scriptsize] {$U_1{=}1$};
%\draw[->, thick] (rec.south east) -- (hat2.north west)
%    node[midway, right=1pt, font=\scriptsize] {$U_2{=}1$};
\draw[->, thick] ([xshift=-1mm]rec.south) -- ([xshift= 1mm]hat1.north)
    node[midway, left=10pt, font=\scriptsize] {$U_1{=}1$};

\draw[->, thick] ([xshift= 1mm]rec.south) -- ([xshift=-1mm]hat2.north)
    node[midway, right=10pt, font=\scriptsize] {$U_2{=}1$};
\end{tikzpicture}%
}

\caption{The considered system model.}
\label{fig:one}
\end{figure}
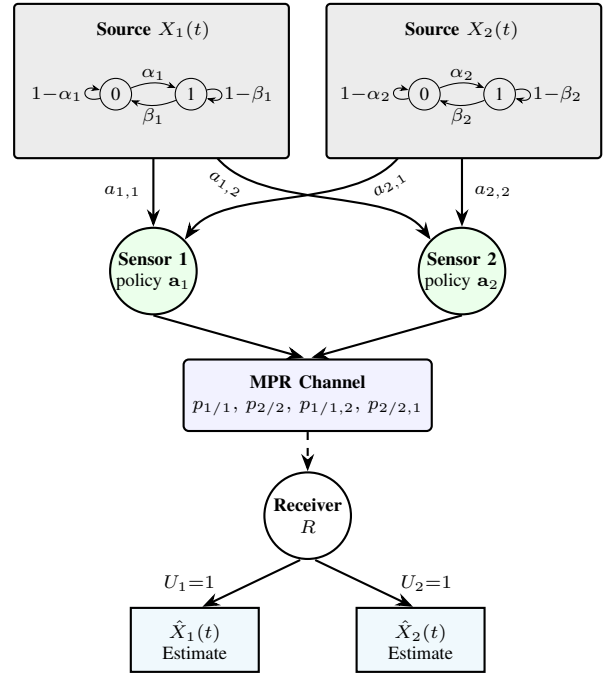

We now derive the effective update probabilities defined in \eqref{eq:sys_qi_def}. Since the sensors select their actions independently under stationary randomized policies, each joint action occurs with probability equal to the product of the corresponding policy probabilities. The effective update probability of source $i$ is then obtained by summing over all joint actions in which at least one successfully decoded packet
contains $X_i$. For source~1, an update occurs in the following cases: sensor~1 transmits $X_1$ while sensor~2 is silent; sensor~2 transmits $X_1$ while sensor~1 is silent; both sensors transmit $X_1$; sensor~1 transmits $X_1$ while sensor~2 transmits $X_2$; or sensor~2 transmits $X_1$ while sensor~1 transmits $X_2$. When both sensors transmit $X_1$, the probability that at least one packet is decoded is $1-(1-p_{1/1,2})(1-p_{2/2,1})$. 

Thus for $j = 1 + (i \bmod 2), i\in\{1,2\}$ we have 

\begin{equation}
\begin{aligned}
%&j = 1 + (i \bmod 2), \qquad i\in\{1,2\},\\
&q_i
=
a_{1,i}a_{2,0}\,p_{1/1}
+
a_{1,0}a_{2,i}\,p_{2/2} \\
&\qquad
+
a_{1,i}a_{2,i}
\Big[1-(1-p_{1/1,2})(1-p_{2/2,1})\Big] \\
&\qquad
+
a_{1,i}a_{2,j}\,p_{1/1,2}
+
a_{1,j}a_{2,i}\,p_{2/2,1}.
\end{aligned}
\label{eq:q1_explicit}
\end{equation}

\section{Analysis} \label{sec:analysis}
\subsection{Real-Time Reconstruction Error (RTE)}
\label{subsec:rte}

We characterize the reconstruction error by studying the joint evolution of
the source $X_i(t)$ and its estimate $\hat X_i(t)$. The pair
$(X_i(t),\hat X_i(t))$ takes values in
\[
\mathcal{S}=\{(0,0),(0,1),(1,0),(1,1)\},
\]
with ordering $s_1=(0,0),\quad s_2=(0,1),\quad s_3=(1,0),\quad s_4=(1,1)$.

We assume that, within each slot, the source first evolves, and then a
successfully decoded update synchronizes the receiver with the current
source state; if no update is decoded, the receiver keeps its previous
estimate. Under this convention, the joint process
$(X_i(t),\hat X_i(t))$ is a four-state Markov chain with transition matrix

$\mathbf{T}_i$, where we define
$Y_i(t)\triangleq (X_i(t),\hat X_i(t))$. Then
\begin{equation}
T_i(k,\ell)
=
\Pr\{Y_i(t+1)=s_\ell \mid Y_i(t)=s_k\}.
\label{eq:T_entry}
\end{equation}

Using \eqref{eq:sys_Pi} and \eqref{eq:sys_sync_or_hold}, direct enumeration
gives

\begin{equation}
\resizebox{0.96\columnwidth}{!}{$
\mathbf{T}_i =
\begin{bmatrix}
1-\alpha_i & 0 & \alpha_i(1-q_i) & \alpha_i q_i \\
q_i(1-\alpha_i) & (1-\alpha_i)(1-q_i) & 0 & \alpha_i \\
\beta_i & 0 & (1-\beta_i)(1-q_i) & q_i(1-\beta_i) \\
\beta_i q_i & \beta_i(1-q_i) & 0 & 1-\beta_i
\end{bmatrix}
$}
\label{eq:T_matrix}
\end{equation}

For example, from state $(0,0)$, the process moves to $(1,0)$ with
probability $\alpha_i(1-q_i)$ when the source changes but no update is
received, and to $(1,1)$ with probability $\alpha_i q_i$ when the source
change is followed by a successful update.

Let $\boldsymbol{\pi}_i
=
\big[
\pi_i(0,0),\,
\pi_i(0,1),\,
\pi_i(1,0),\,
\pi_i(1,1)
\big]$
denote the stationary distribution of the joint chain. Then
\begin{equation}
\boldsymbol{\pi}_i = \boldsymbol{\pi}_i \mathbf{T}_i,
\qquad
\sum_{(x,\hat x)\in\mathcal{S}} \pi_i(x,\hat x) = 1.
\label{eq:stationary_equations}
\end{equation}
Following \cite{10409276}, the instantaneous reconstruction error is
\begin{equation}
E_i(t) \triangleq \mathbf{1}\{X_i(t) \neq \hat{X}_i(t)\}.
\label{eq:reconstruction_error}
\end{equation}
Thus, the steady-state real-time reconstruction error (RTE) is
\begin{equation}
E_i
\triangleq
\Pr\big(X_i(t)\neq \hat X_i(t)\big)
=
\pi_i(0,1) + \pi_i(1,0).
\label{eq:Ei_def}
\end{equation}
Solving \eqref{eq:stationary_equations} and substituting into
\eqref{eq:Ei_def} yields
\begin{equation}
E_i
=
\frac{2\alpha_i\beta_i(1-q_i)}
{(\alpha_i+\beta_i)\big[(\alpha_i+\beta_i) - q_i(\alpha_i+\beta_i-1)\big]}.
\label{eq:Ei_closed_form}
\end{equation}
The detailed derivation is omitted due to space limitations but is provided in the online appendix~\ref{app:Ei_derivation}.
Equivalently, defining $\lambda_i\triangleq 1-\alpha_i-\beta_i$, we obtain
\begin{equation}
E_i
=
\frac{2\alpha_i\beta_i(1-q_i)}
{(1-\lambda_i)\big(1-\lambda_i(1-q_i)\big)}.
\label{eq:Ei_closed_form_lambda}
\end{equation}

For the two-source system, we consider the weighted reconstruction error
\begin{equation}
E
= w_1 E_1 + w_2 E_2,
\qquad w_1,w_2\ge 0,
\label{eq:weighted_total_error}
\end{equation}
where $w_i$ captures the importance of source $i$.
\begin{remark}

The expression in \eqref{eq:Ei_closed_form} is stated for
$q_i\in(0,1]$. For $0<q_i<1$, the joint Markov chain in
\eqref{eq:T_matrix} is irreducible. If $q_i=1$, then $E_i=0$, since the
receiver is synchronized in every slot. As $q_i\rightarrow 0$,
\[
E_i
\xrightarrow[q_i\rightarrow 0]{}
\frac{2\alpha_i\beta_i}{(\alpha_i+\beta_i)^2},
\]
which is the mismatch probability between two independent samples drawn
from the stationary distribution of $X_i(t)$. When $q_i=0$, no update
is ever received, so $\hat X_i(t)$ remains fixed at its initial value and
the joint chain is not irreducible. The parameter
$\lambda_i=1-\alpha_i-\beta_i$ captures the temporal correlation of the
source: $\lambda_i>0$ corresponds to positively correlated dynamics,
$\lambda_i=0$ to the memoryless case, and $\lambda_i<0$ to negatively
correlated dynamics. 

Fig.~\ref{fig:RCE_alpha_beta} illustrates the effect of
$(\alpha_i,\beta_i)$ on the RTE.

\end{remark}

\begin{figure}[t]
\centering
\includegraphics[width=0.65\linewidth]{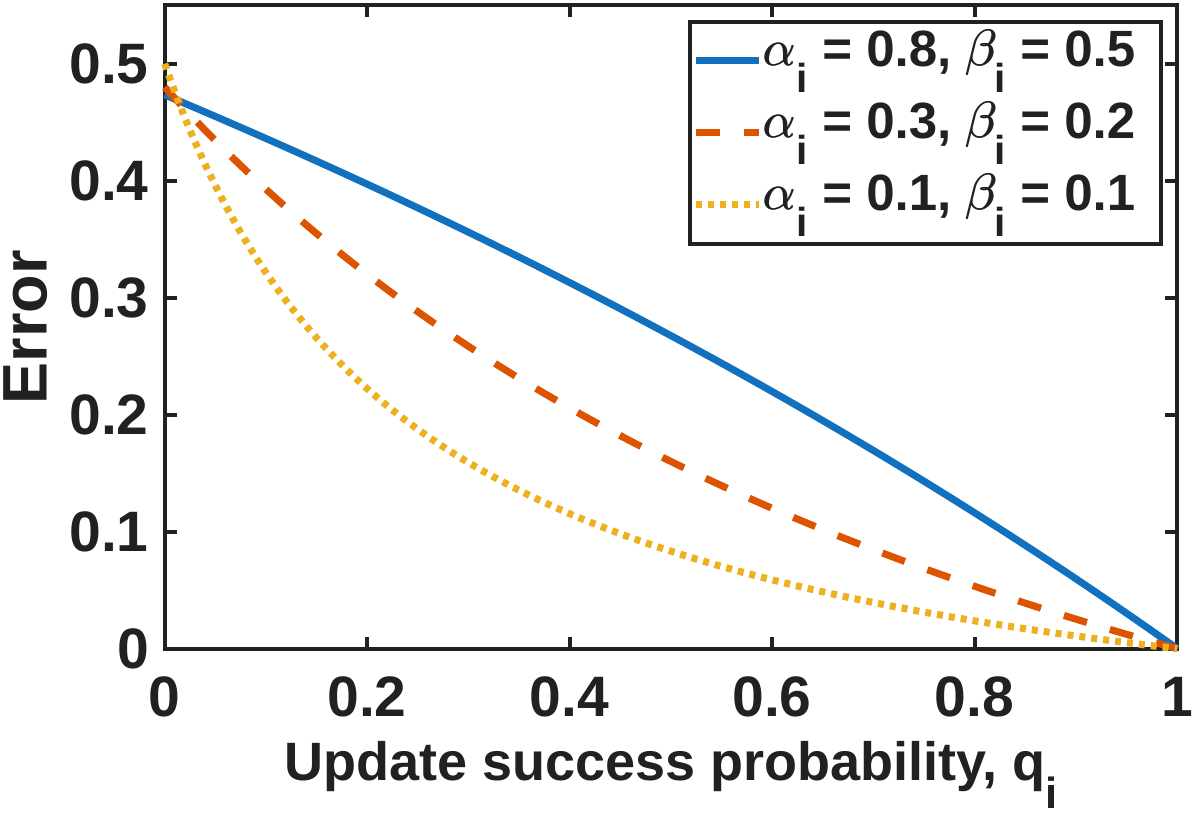}
\caption{RTE for different
$(\alpha_i,\beta_i)$ as a function of $q_i$.}
\label{fig:RCE_alpha_beta}
\end{figure}

\subsection{Cost of Actuation Error}
\label{subsec:cae}

We next assign costs to the erroneous reconstruction states. Following
\cite{10409276}, the generic cost of actuation error (CAE) is
\begin{equation}
\bar{C}_{i}
\triangleq
\sum_{x \neq \hat{x}} C_i^{x,\hat{x}}\,\pi_i(x,\hat{x}),
\label{eq:CAE_generic}
\end{equation}
where $C_i^{x,\hat{x}}$ is the cost incurred for source $i$ when the true
state is $x$ while the reconstructed state is $\hat{x}\neq x$.

Let $C_i^{0,1}$ and $C_i^{1,0}$ denote the costs of the
errors $(X_i,\hat X_i)=(0,1)$ and $(1,0)$, respectively. Then
\begin{equation}
\overline{\mathcal C}_i
=
C_i^{0,1}\,\pi_i(0,1)
+
C_i^{1,0}\,\pi_i(1,0).
\label{eq:CAE_per_source}
\end{equation}
For $q_i>0$, the stationary balance equations of the four-state chain imply
\begin{equation}
\zeta_i \triangleq \pi_i(0,1)=\pi_i(1,0).
\label{eq:zeta_def}
\end{equation}
Therefore,
\begin{equation}
E_i=\pi_i(0,1)+\pi_i(1,0)=2\zeta_i,
\qquad
\zeta_i=\frac{E_i}{2}.
\label{eq:Ei_zeta_relation}
\end{equation}
Using \eqref{eq:Ei_closed_form}, we obtain
\begin{equation}
\zeta_i
=
\frac{\alpha_i\beta_i(1-q_i)}
{(\alpha_i+\beta_i)\big[(\alpha_i+\beta_i)-q_i(\alpha_i+\beta_i-1)\big]},
\label{eq:zeta_closed_form}
\end{equation}
or equivalently, with $\lambda_i\triangleq1-\alpha_i-\beta_i$,
\begin{equation}
\zeta_i
=
\frac{\alpha_i\beta_i(1-q_i)}
{(1-\lambda_i)\big(1-\lambda_i(1-q_i)\big)}.
\label{eq:zeta_closed_form_lambda}
\end{equation}
Substituting \eqref{eq:zeta_def} into \eqref{eq:CAE_per_source} gives
\begin{equation}
\overline{\mathcal C}_i
=
\big(C_i^{0,1}+C_i^{1,0}\big)\zeta_i.
\label{eq:CAE_closed_form}
\end{equation}
Hence,
\begin{equation}
\overline{\mathcal C}_i
=
\frac{\big(C_i^{0,1}+C_i^{1,0}\big)\alpha_i\beta_i(1-q_i)}
{(1-\lambda_i)\big(1-\lambda_i(1-q_i)\big)}.
\label{eq:CAE_closed_form_lambda}
\end{equation}
The weighted total CAE is
\begin{equation}
\overline{\mathcal C}
=
w_1\,\overline{\mathcal C}_1
+
w_2\,\overline{\mathcal C}_2,
\qquad w_1,w_2\ge 0.
\label{eq:Ctot}
\end{equation}

\begin{remark}
Although the directional costs $C_i^{0,1}$ and $C_i^{1,0}$ may be different, the two mismatch states have equal steady-state probabilities under the considered binary Markov model and synchronize-or-hold estimator. Consequently,
\begin{equation}
\overline{\mathcal C}_i
=
\big(C_i^{0,1}+C_i^{1,0}\big)\zeta_i
=
\frac{C_i^{0,1}+C_i^{1,0}}{2}E_i.
\label{eq:cae_scaled_rte}
\end{equation}
Thus, for each binary source under the stationary randomized policies considered here, the CAE is proportional to the corresponding RTE, with proportionality factor $(C_i^{0,1}+C_i^{1,0})/2$. Consequently, weighted CAE minimization is equivalent to weighted RTE minimization with modified source weights. This reduction does not rely on symmetric actuation costs; it follows from the equality of the two steady-state mismatch probabilities. 

For multi-state Markov sources, or for more general dynamic policies whose update decisions depend on the current source state or on the particular mismatch state $(X_i,\hat X_i)$, this simplification does not generally hold, since different mismatch states may have different stationary probabilities.
\end{remark}

\section{Optimization under Sampling Constraints}
\label{sec:optimization}

We optimize the stationary randomized sampling policy with respect to
the weighted reconstruction objective in \eqref{eq:weighted_total_error}.
For each sensor $k\in\{1,2\}$, let $\Gamma_k\in(0,1]$ denote the maximum
allowable sampling rate, i.e., the maximum probability with which sensor
$k$ is allowed to transmit in a slot. Since transmission occurs
whenever action $1$ or $2$ is selected, the sampling-budget constraint is
$a_{k,1}+a_{k,2}\le \Gamma_k, k\in\{1,2\}$.
The constrained weighted RTE minimization problem is formulated as\footnote{For the binary-source model, the weighted CAE minimization
problem has the same feasible set and is equivalent to weighted RTE
minimization with modified source weights
\[
\tilde w_i
=
w_i\frac{C_i^{0,1}+C_i^{1,0}}{2},
\qquad i\in\{1,2\}.
\]
Thus, the results apply to CAE minimization after
this weight transformation.}
\begin{align}
\min_{\{a_{k,j}\}} \quad & E
\label{eq:rte_optimization}\\
\text{s.t.}\quad
& a_{k,j} \ge 0, \qquad k\in\{1,2\},\; j\in\{0,1,2\}, \label{eq:opt_nonneg}\\
& a_{k,0}+a_{k,1}+a_{k,2} = 1, \qquad k\in\{1,2\}, \label{eq:opt_simplex}\\
& a_{k,1}+a_{k,2} \le \Gamma_k, \qquad k\in\{1,2\}. \label{eq:opt_budget}
\end{align}

Both objectives depend on the policy through the effective update
probabilities $q_i(\mathbf a_1,\mathbf a_2)$, given by
\eqref{eq:q1_explicit}. Since these expressions
contain bilinear terms of the form $a_{1,j}a_{2,j'}$, the optimization
problem in \eqref{eq:rte_optimization} is, in general, nonconvex.

We now derive a structural result for the sampling-constrained problem in
\eqref{eq:rte_optimization}--\eqref{eq:opt_budget}. For each sensor
$k\in\{1,2\}$, the feasible policy set is

\begin{equation}
\begin{aligned}
\mathcal P_k(\Gamma_k)
\triangleq
\bigl\{\mathbf a_k\in\mathbb R_{\ge0}^3:\;&
a_{k,0}+a_{k,1}+a_{k,2}=1,\\
&
a_{k,1}+a_{k,2}\le \Gamma_k
\bigr\}.
\end{aligned}
\label{eq:Pk_gamma}
\end{equation}

Equivalently, since $a_{k,0}=1-a_{k,1}-a_{k,2}$, the feasible set in the
$(a_{k,1},a_{k,2})$ plane is the triangular region $a_{k,1}\ge0, a_{k,2}\ge0, a_{k,1}+a_{k,2}\le \Gamma_k$.
Then, the extreme points of $\mathcal P_k(\Gamma_k)$ are
\begin{equation}
\mathcal V_k
=
\left\{
(1,0,0),\;
(1-\Gamma_k,\Gamma_k,0),\;
(1-\Gamma_k,0,\Gamma_k)
\right\}.
\label{eq:Vk_gamma}
\end{equation}
The constrained feasible set is $\mathcal F_\Gamma
=
\mathcal P_1(\Gamma_1)\times \mathcal P_2(\Gamma_2).$

From \eqref{eq:q1_explicit}, each
$q_i(\mathbf a_1,\mathbf a_2)$ is bilinear in the two policy blocks.
Hence, for fixed $\mathbf a_j$, it is affine in $\mathbf a_i$. When
\begin{equation}
\lambda_1\le 0,\qquad \lambda_2\le 0,
\label{eq:lambda_nonpositive_constrained}
\end{equation}
 
it is straightforward to verify that $E_i(q_i)$ is concave in $q_i$.
Thus, the mappings
\begin{equation}
\mathbf a_1 \mapsto E_i\big(q_i(\mathbf a_1,\mathbf a_2)\big),
\qquad
\mathbf a_2 \mapsto E_i\big(q_i(\mathbf a_1,\mathbf a_2)\big)
\label{eq:block_concavity_components_constrained}
\end{equation}
are concave in each block. Consequently,
\begin{equation}
E(\mathbf a_1,\mathbf a_2)
=
w_1E_1\big(q_1(\mathbf a_1,\mathbf a_2)\big)
+
w_2E_2\big(q_2(\mathbf a_1,\mathbf a_2)\big)
\label{eq:Etot_special_constrained}
\end{equation}
is concave in $\mathbf a_1$ for fixed $\mathbf a_2$, and concave in
$\mathbf a_2$ for fixed $\mathbf a_1$.

\begin{theorem}
\label{thm:caseA_constrained}
If $\lambda_1\le 0$ and $\lambda_2\le 0$, then at least one global minimizer
of the sampling-constrained RTE problem over $\mathcal F_\Gamma$ is
attained at a pair of constrained vertices. Equivalently,
\begin{equation}
E^\star
=
\min_{\mathbf v_1\in\mathcal V_1,\;\mathbf v_2\in\mathcal V_2}
E(\mathbf v_1,\mathbf v_2).
\label{eq:Etot_vertex_constrained}
\end{equation}
Thus, a global optimum can be found by evaluating only nine candidate
policy pairs.
\end{theorem}

\begin{IEEEproof}
Fix $\mathbf a_2\in\mathcal P_2(\Gamma_2)$. Since
$E(\mathbf a_1,\mathbf a_2)$ is concave in $\mathbf a_1$, it admits at
least one minimizer over the compact convex polytope
$\mathcal P_1(\Gamma_1)$ at an extreme point. Hence,
\begin{equation}
\min_{\mathbf a_1\in\mathcal P_1(\Gamma_1)}
E(\mathbf a_1,\mathbf a_2)
=
\min_{\mathbf v_1\in\mathcal V_1}
E(\mathbf v_1,\mathbf a_2).
\label{eq:first_block_min_constrained}
\end{equation}
For each fixed $\mathbf v_1\in\mathcal V_1$, the function
$E(\mathbf v_1,\mathbf a_2)$ is concave in $\mathbf a_2$. Therefore, it
admits at least one minimizer over $\mathcal P_2(\Gamma_2)$ at an extreme
point
\begin{equation}
\min_{\mathbf a_2\in\mathcal P_2(\Gamma_2)}
E(\mathbf v_1,\mathbf a_2)
=
\min_{\mathbf v_2\in\mathcal V_2}
E(\mathbf v_1,\mathbf v_2).
\label{eq:second_block_min_constrained}
\end{equation}
Combining \eqref{eq:first_block_min_constrained} and
\eqref{eq:second_block_min_constrained} proves
\eqref{eq:Etot_vertex_constrained}.
\end{IEEEproof}

For example, the constrained vertex pair
\[
\mathbf v_1=(1-\Gamma_1,\Gamma_1,0),
\qquad
\mathbf v_2=(1,0,0)
\]
corresponds to sensor~1 using its full budget for source~1 while sensor~2
remains silent, and gives $(q_1,q_2)=(\Gamma_1 p_{1/1},0)$.
Similarly, the vertex pair
\[
\mathbf v_1=(1-\Gamma_1,\Gamma_1,0),
\qquad
\mathbf v_2=(1-\Gamma_2,0,\Gamma_2)
\]
corresponds to sensor~1 allocating its full budget to source~1 and
sensor~2 allocating its full budget to source~2. In this case,

\begin{equation}
\begin{aligned}
q_i
&=
\Gamma_i(1-\Gamma_j)p_{i/i}
+
\Gamma_i\Gamma_j p_{i/i,j},\\
j&=1+(i\bmod 2),\quad i\in\{1,2\}.
\end{aligned}
\label{eq:q_vertex_example}
\end{equation}

The remaining seven candidates are obtained by substituting all pairs $(\mathbf v_1,\mathbf v_2)\in\mathcal V_1\times\mathcal V_2$ into \eqref{eq:q1_explicit}. Thus, under \eqref{eq:lambda_nonpositive_constrained}, the constrained optimization reduces to evaluating $E(\mathbf v_1,\mathbf v_2)$ over the nine vertex pairs in $\mathcal V_1\times\mathcal V_2$.

\begin{remark}
In this regime $\lambda_i\le0$, the per-source error functions are concave in the update probabilities, and at least one constrained vertex policy is globally
optimal. When at least one $\lambda_i$ is positive, this concavity is lost and an interior randomized policy may become optimal. This case is omitted due to space limitations.

\end{remark}

\section{Numerical Results and Discussion}
\label{sec:results}

In this section, we evaluate the optimized stationary randomized policy. We compare the optimized policy with four baselines: (i) a \emph{random policy}, where each sensor splits its sampling budget uniformly between the two sources;\\ (ii) \emph{Greedy Source~1}, where both sensors allocate their entire available budget to source~$X_1$; (iii) \emph{Greedy Source~2}, where both sensors allocate their entire available budget to source~$X_2$; and (iv) a \emph{time-division multiple access (TDMA) policy}
\footnote{In the TDMA baseline, at most one sensor is allowed to transmit in each slot. We optimize over the long-term time fractions  \(\tau_{11},\tau_{12},\tau_{21},\tau_{22}\), where \(\tau_{ki}\) denotes the  fraction of slots in which sensor \(k\) transmits source \(i\). The same per-sensor sampling budgets are used as in the MPR case, i.e.,
\(\tau_{11}+\tau_{12}\leq \Gamma_1\) and 
\(\tau_{21}+\tau_{22}\leq \Gamma_2\), with the additional TDMA  orthogonality constraint 
\(\tau_{11}+\tau_{12}+\tau_{21}+\tau_{22}\leq 1\). 
The time fractions are optimized and are not fixed to an equal split in order to achieve better performance.}, 
where the sensors are orthogonalized in time so that at most one sensor transmits in each slot. When scheduled, sensor~$i$ uses its available sampling budget to transmit source~$X_i$.
\subsection{Impact of the Sampling Budget}
\label{subsec:gamma_results}

Fig.~\ref{fig:optimization_results} shows the total reconstruction error $E$ as a function of the common sampling budget $\Gamma$. We use 
$ (\alpha_1,\beta_1)=(0.8,0.6), (\alpha_2,\beta_2)=(0.3,0.2)$, representing a highly dynamic source and a moderately dynamic source, respectively. The MPR parameters are  $p_{1/1}=0.9,p_{2/2}=0.85,
p_{1/1,2}=0.6, p_{2/2,1}=0.55$, with weights $w_1=w_2=0.5$ and $\Gamma\in[0.01,0.95]$.

\begin{figure}[t]
\centering
\includegraphics[width=0.7\linewidth]{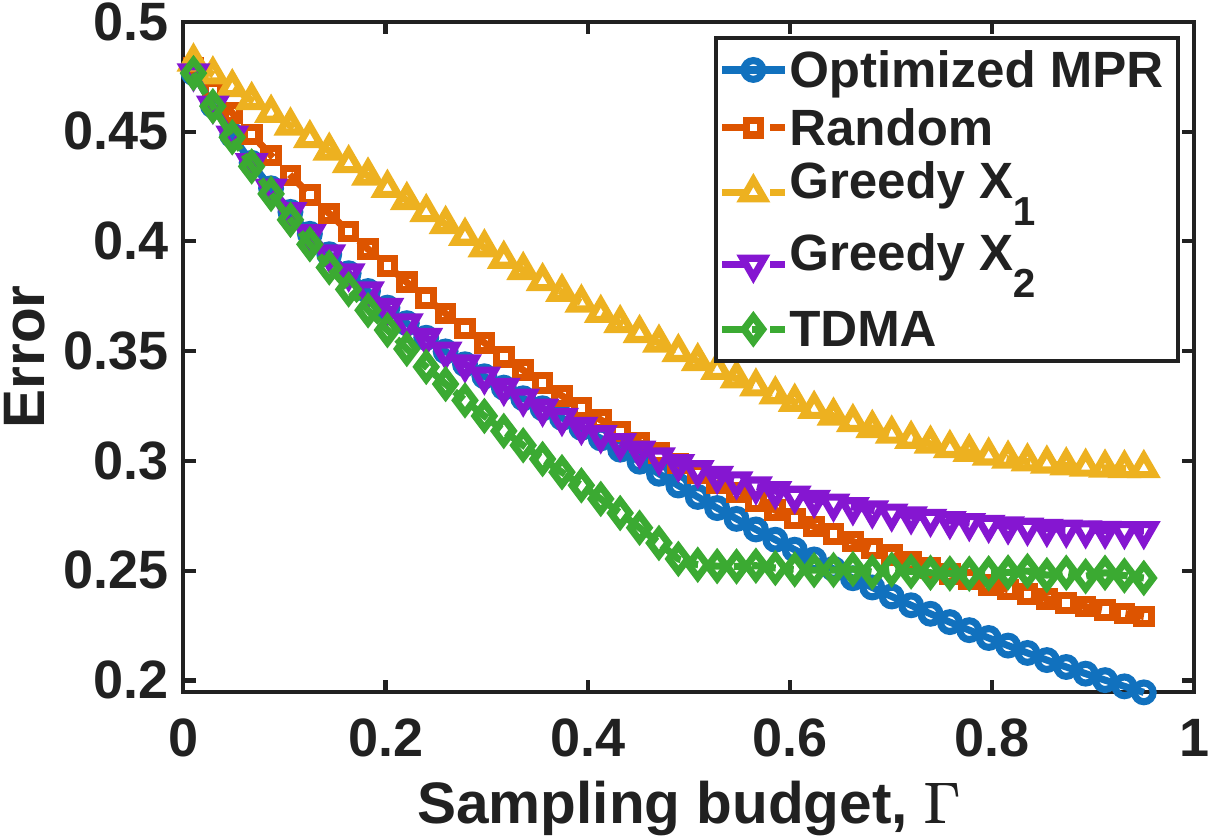}
\caption{Total reconstruction error versus sampling budget $\Gamma$.}
\label{fig:optimization_results}
\end{figure}

As expected, increasing $\Gamma$ reduces $E$ for all policies, since more
frequent sampling increases the update probabilities. The optimized policy
achieves a consistently lower reconstruction error than the random and
greedy baselines, especially at moderate and large budgets. This shows that
allocating all resources to a single source is generally suboptimal when
the sources have different dynamics and the update probabilities are
coupled through the MPR channel. Instead, the optimized policy balances the
sampling resources across sources according to their dynamics, channel
reliability, and contribution to the weighted reconstruction error.

The comparison with TDMA illustrates the role of simultaneous
transmissions. At low sampling budgets, TDMA is advantageous because scarce
transmission opportunities avoid interference. At high budgets, however,
the MPR policy becomes superior because simultaneous transmissions can be
exploited to deliver multiple updates within the same slot, rather than
restricting transmissions through strict orthogonalization.

\subsection{Effect of Semantic Weights}
\label{subsec:semantic_weight_results}

We next study the effect of the semantic weight $w_2$ on the optimized
policy, with $w_1=1-w_2$. The parameters are
$ (\alpha_1,\beta_1)=(0.8,0.1),
(\alpha_2,\beta_2)=(0.4,0.2)$, and
$p_{1/1}=0.9, p_{2/2}=0.85,
p_{1/1,2}=0.82, p_{2/2,1}=0.78$.
Both sensors satisfy $\Gamma_1=\Gamma_2=0.9$.

\begin{figure}[t]
    \centering
    \includegraphics[width=0.68\linewidth]{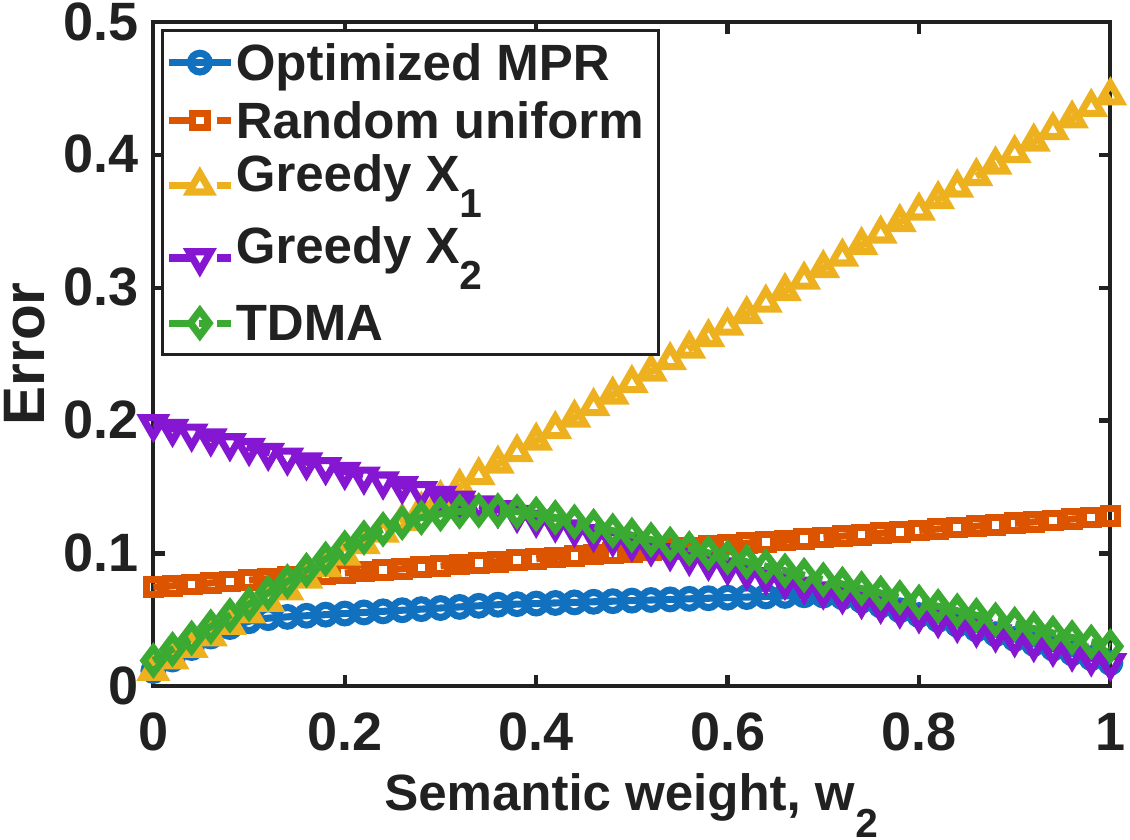}
    \caption{Total reconstruction error $E(w_2)$ versus semantic weight.}
    \label{fig:Etot_vs_w2}
\end{figure}

Fig.~\ref{fig:Etot_vs_w2} shows that the optimized policy adapts to the
semantic priority of the sources. When $w_2$ is small, the objective is
dominated by source~$1$, and the optimized performance approaches the
Greedy Source~1 baseline. As $w_2$ increases, the policy reallocates
sampling resources toward source~$2$. In the intermediate regime, the
optimized curve exhibits a mild non-monotonic behavior, reflecting the
tradeoff between source dynamics, semantic weights, and MPR coupling.
Unlike the greedy baselines, which commit the sampling budget to one
source, the optimized randomized policy balances the two sources and
achieves the lowest reconstruction error over the considered range. 

The TDMA curve lies between the two greedy baselines because it orthogonalizes transmissions while allocating time across the sources.

In this strong-MPR regime, TDMA is less efficient because it prevents simultaneous transmissions that could otherwise be decoded reliably. Hence, with high sampling budgets, MPR better exploits the channel by allowing concurrent updates, leading to lower reconstruction error.

\section{Conclusion}
\label{sec:conclusion}

This paper studied real-time reconstruction and actuation for two binary Markov sources sharing an MPR channel. We derived closed-form expressions for the RTE and CAE in terms of the source dynamics and the effective
update probabilities, and characterized how these update probabilities depend on the stationary randomized sampling policy and the MPR reception model. For binary sources, we showed that CAE minimization is equivalent to weighted RTE minimization with modified source weights. We then formulated a sampling-constrained weighted-error minimization problem and identified a source-dynamics regime in which a constrained vertex policy is globally optimal. Numerical results showed that optimized randomized policies outperform random, greedy, and time-sharing baselines, and highlighted the impact of source dynamics, channel
conditions, semantic weights, and MPR coupling on sampling-resource allocation. Future work may consider larger setups, and adaptive sampling policies.

\appendix
\section{Derivation of the Closed-Form Reconstruction Error}
\label{app:Ei_derivation}

In this appendix, we provide the algebraic steps leading to the
closed-form expression of the steady-state real-time reconstruction
error in~\eqref{eq:Ei_closed_form}.

Let
\[
\pi_{00}\triangleq \pi_i(0,0),\qquad
\pi_{01}\triangleq \pi_i(0,1),
\]
\[
\pi_{10}\triangleq \pi_i(1,0),\qquad
\pi_{11}\triangleq \pi_i(1,1).
\]
so that $\boldsymbol{\pi}_i = [\pi_{00},\pi_{01},\pi_{10},\pi_{11}]$.

From the stationary condition $\boldsymbol{\pi}_i = \boldsymbol{\pi}_i \mathbf{T}_i$

and the transition matrix in~\eqref{eq:T_matrix}, the balance
equations corresponding to the two mismatch states \((0,1)\) and
\((1,0)\) are
\begin{equation}
\pi_{01}
=
(1-q_i)\big[(1-\alpha_i)\pi_{01}+\beta_i\pi_{11}\big],
\label{eq:app_balance_01}
\end{equation}
\begin{equation}
\pi_{10}
=
(1-q_i)\big[\alpha_i\pi_{00}+(1-\beta_i)\pi_{10}\big].
\label{eq:app_balance_10}
\end{equation}

Next, note that the marginal distribution of \(X_i(t)\) must coincide
with the stationary distribution of the underlying binary Markov source.
Hence,

\begin{equation}
\pi_{00}=\frac{\beta_i}{\alpha_i+\beta_i}-\pi_{01},
\qquad
\pi_{11}=\frac{\alpha_i}{\alpha_i+\beta_i}-\pi_{10}.
\label{eq:app_pi00_pi11}
\end{equation}

Substituting~\eqref{eq:app_pi00_pi11} into~\eqref{eq:app_balance_01}
gives
\begin{equation}
\big[1-(1-q_i)(1-\alpha_i)\big]\pi_{01}
+\beta_i(1-q_i)\pi_{10}
=
\frac{\alpha_i\beta_i(1-q_i)}{\alpha_i+\beta_i}.
\label{eq:app_linear_01}
\end{equation}

Similarly, substituting~\eqref{eq:app_pi00_pi11} into
\eqref{eq:app_balance_10} yields

\begin{equation}
\alpha_i(1-q_i)\pi_{01}
+\big[1-(1-q_i)(1-\beta_i)\big]\pi_{10}
=
\frac{\alpha_i\beta_i(1-q_i)}{\alpha_i+\beta_i}.
\label{eq:app_linear_10}
\end{equation}

Subtracting~\eqref{eq:app_linear_10} from~\eqref{eq:app_linear_01},
we obtain
\begin{equation}
q_i(\pi_{01}-\pi_{10})=0,
\end{equation}
which implies
\begin{equation}
\pi_{01}=\pi_{10}.
\label{eq:app_pi01_eq_pi10}
\end{equation}

Let
\[
\zeta_i \triangleq \pi_{01}=\pi_{10}.
\]
Substituting~\eqref{eq:app_pi01_eq_pi10} into
\eqref{eq:app_linear_01} gives
\[
\Big(1-(1-q_i)(1-\alpha_i)+\beta_i(1-q_i)\Big)\zeta_i
=
\frac{\alpha_i\beta_i(1-q_i)}{\alpha_i+\beta_i}.
\]

Finally, since $E_i=\pi_{01}+\pi_{10}=2\zeta_i$, we obtain
\begin{equation}
E_i
=
\frac{2\alpha_i\beta_i(1-q_i)}
{(\alpha_i+\beta_i)\big[(\alpha_i+\beta_i)-q_i(\alpha_i+\beta_i-1)\big]}.
\end{equation}

\bibliographystyle{IEEEtran}
\bibliography{Refs}
\end{document}